%% file: IEEE-conference-template-062824.tex
\documentclass[conference]{IEEEtran}
\IEEEoverridecommandlockouts

\usepackage[utf8]{inputenc} 
\usepackage[T1]{fontenc}
\usepackage{url}
\usepackage{ifthen}
\usepackage{cite}
\usepackage[cmex10]{amsmath} 
\usepackage{color}
\usepackage{graphicx}
\usepackage{multirow} 
\usepackage{amsthm}
\usepackage{amsmath}
\usepackage{amssymb}
\usepackage{amsfonts}
\usepackage{caption}
\usepackage{cancel}
\usepackage{mathtools}

\mathtoolsset{showonlyrefs} 

\theoremstyle{definition}
\newtheorem{theorem}{Theorem}
\newtheorem*{theorem*}{Theorem}

\newtheorem*{definition*}{Definition}
\newtheorem{dfn}[theorem]{Definition}

\newtheorem{prp}[theorem]{Proposition}
\newtheorem{thr}[theorem]{Theorem}
\newtheorem{lmm}[theorem]{Lemma}

\newcommand{\dddots}{\mathbin{\rotatebox[origin=c]{15}{$\ddots$}}}
\newcommand{\iiddots}{\mathbin{\rotatebox[origin=c]{65}{$\ddots$}}}

\def\=def{\overset{\text{\small def}}{=}}
\newcommand{\defeq}{\overset{\text{\small def}}{=}}
\def\Zb{\mathbb{Z}}
\def\HH{\hat{H}}
\def\Mc{\mathcal{M}}
\def\Fc{\mathcal{F}}
\newcommand{\T}{\mathsf{T}}
\def\Fb{\mathbb{F}}

\def\BibTeX{{\rm B\kern-.05em{\sc i\kern-.025em b}\kern-.08em
    T\kern-.1667em\lower.7ex\hbox{E}\kern-.125emX}}
\begin{document}

\title{
    {Explicit Construction of Quantum Quasi-Cyclic Low-Density Parity-Check Codes with Column Weight 2 and Girth 12}
}

\author{\IEEEauthorblockN{Daiki Komoto and Kenta Kasai}
\IEEEauthorblockA{\textit{Department of Information and Communications Engineering} \\
\textit{Institute of Science Tokyo}\\
2--12--1 Ookayama, Meguro-ku, 152--8550, Tokyo, Japan \\
\{komoto.d.aa, kasai.k.ab\}@m.titech.ac.jp}
}

\maketitle

\begin{abstract}
    This study proposes an explicit construction method for quantum quasi-cyclic low-density parity-check (QC-LDPC) codes with a girth of 12. 
    The proposed method designs parity-check matrices that maximize the girth while maintaining an orthogonal structure suitable for quantum error correction. 
    By utilizing algebraic techniques, short cycles are eliminated, which improves error correction performance. 
    Additionally, this method is extended to non-binary LDPC codes and spatially-coupled LDPC codes, demonstrating that both the girth and orthogonality can be preserved. 
    The results of this study enable the design of high-performance quantum error-correcting codes without the need for random search. 
\end{abstract}


\section{{Introduction}}


low-density parity-check (LDPC) codes are a class of error-correcting codes that allow efficient decoding via the Sum-Product (SP) algorithm.
However, it is well known that the presence of short cycles in the Tanner graph defining an LDPC code significantly degrades the error-correction performance of the SP algorithm \cite{mackay1999good}. 
To mitigate this issue, considerable research has been devoted to designing parity-check matrices that maximize the minimum cycle length, referred to as the girth, in the associated Tanner graph. A common approach for constructing LDPC codes with large girth involves randomly generating candidate matrices and selecting those with the largest observed minimum cycle length.


Quasi-cyclic (QC-) LDPC codes \cite{fossorier2004quasicyclic} are 
regular LDPC codes—whose parity-check matrices have constant column and row weights 
\( J \) and \( L \), respectively—constructed using circulant permutation matrices (CPMs), 
which provide structured algebraic properties suitable for hardware implementation.


The girth of QC-LDPC codes is upper-bounded by 12 \cite{fossorier2004quasicyclic}. 
Although many works have aimed to explicitly construct QC-LDPC codes 
while avoiding short cycles \cite{lally2007explicit,tasdighi2016symmetrical,zhang2014deterministic,zhang2024explicit,wang2014explicit}, 
achieving larger girth typically required random search. 
In contrast, a deterministic construction that achieves girth 12 eliminates this step entirely.

The Calderbank-Shor-Steane (CSS) codes \cite{calderbank1996good,steane1996error} form a class of quantum error-correcting codes that are constructed using two classical codes. A fundamental requirement for a CSS code is that its check matrices, \( H_X \) and \( H_Z \), must satisfy the orthogonality condition:  
$H_X H_Z^\T = O$. 
When CSS codes are constructed using classical LDPC codes as component codes, they are referred to as quantum LDPC (QLDPC) codes. Similar to their classical counterparts, QLDPC codes employ the SP algorithm for efficient decoding. Consequently, it is crucial to design the underlying orthogonal classical LDPC matrices while eliminating short cycles to enhance the decoding performance. 
To date, quantum QC-LDPC codes with a girth greater than 4 have been proposed \cite{hagiwara2007quantum}, but due to the difficulty of constructing QC-LDPC matrix pairs while maintaining orthogonality, it has not been possible to increase the girth in the quantum case.

Hagiwara and Imai proposed a construction of orthogonal QC-LDPC matrix pairs with girth greater than 4 in~\cite{hagiwara2007quantum}.  
This construction has been extended to non-binary~\cite{kasai2011quantum} and spatially coupled~\cite{hagiwara2011spatially} formats to improve the performance of quantum LDPC codes; however, none of these constructions have explicitly achieved girth 12.  
Furthermore,~\cite{komoto2024quantum} showed that the finite field extension method proposed in~\cite{kasai2011quantum} is not always applicable to general orthogonal QC-LDPC matrix pairs and clarified the conditions under which it can be applied.  
While these conditions enable non-binary extensions, no SC-NB-QC-LDPC codes with \( J = 2 \) and girth 12 have been proposed to date.


Column weight \( J = 2 \) is particularly important because 
(i) the girth upper bound for quantum QC-LDPC codes is limited to 6 when \( J \ge 3 \)~\cite{amirzade2024girth}, 
(ii) NB-LDPC codes generally perform best at \( J = 2 \)~\cite{poulliat2008design}, 
and (iii) the extension method to orthogonal NB-LDPC matrix pairs in~\cite{kasai2011quantum} applies specifically to \( J = 2 \).

In this study, we first provide an explicit construction method for QC-LDPC matrices with column weight \( J = 2 \), which achieves the upper bound of girth 12, in Subsection \ref{sec:Explicit Construction of Classical QC-LDPC Matrices with Girth 12}. 
In this subsection, we prove that the constructed matrices achieve a girth of 12. 
Furthermore, in Subsection~\ref{sec:quantum-QC-LDPC-matrices} and subsequent sections, we demonstrate that the same approach can be applied to prove girth 12 for other constructions as well.

Next, in Subsection \ref{sec:quantum-QC-LDPC-matrices}, 
we propose a construction method for orthogonal QC-LDPC matrix pairs for quantum LDPC codes. 
Specifically, we describe how to construct orthogonal matrix pairs using the matrices 
introduced in Subsection~\ref{sec:Explicit Construction of Classical QC-LDPC Matrices with Girth 12}. 
This construction method for orthogonal QC-LDPC matrix pairs is based on the method 
proposed by Hagiwara and Imai for orthogonal QC-LDPC matrix pairs \cite{hagiwara2007quantum}.


Furthermore, we show in Section~\ref{sec:finite-field-extension} and Section~\ref{sec:spatial-coupling} that the proposed orthogonal QC-LDPC matrix pairs can be extended to NB-LDPC and SC-LDPC matrix pairs while preserving both orthogonality and girth 12.  
The extension to NB-LDPC matrices is based on the method proposed by Kasai et al.~\cite{kasai2011quantum}, 
and the extension to SC-LDPC matrices is based on the method by Hagiwara et al.~\cite{hagiwara2011spatially}.

The performance evaluation results are almost the same as those presented in \cite{komoto2024quantum}. 




\section{{QC-LDPC Matrices}}\label{sec:QC-LDPC-matrices}
In this paper, we define 
$\Zb_P \defeq \Zb / P\Zb = \{0,\cdots,P-1\}, L/2+\Zb_{L/2} = \{L/2,\cdots,L-1\}$. 
One approach to constructing regular LDPC matrices is to arrange multiple permutation matrices in a structured layout.
Since permutation matrices have a row weight and column weight of 1, a parity-check matrix constructed by arranging \( J \times L \) permutation matrices of size \( P \) results in a \((J,L,P)\)-regular LDPC matrix. 
An arbitrary permutation matrix \( F \) of size \( P \) can be defined using a permutation mapping \( f \), which is a bijective function \( \Zb_P \to \Zb_P \), such that 
$(F)_{f(i),i}=1$. 
We denote the set of permutation matrices of size \( P \) as \( \Mc_P \) and the set of permutation mappings as \( \Fc_P \). Moreover, for a given permutation matrix \( F \in \Mc_P \) and its corresponding permutation mapping \( f \in \Fc_P \), we express their relationship as \( F \sim f \).

A simple example of a permutation matrix is the circulant permutation matrix (CPM). 
A CPM is a permutation matrix defined by a permutation mapping of a specific algebraic form, referred to as a CPM mapping.

\begin{dfn}
    [CPM and QC-LDPC Matrices]  
    For any \( b \in \Zb_P \), the function \( f(x) = x + b \pmod{P} \) defines a permutation mapping.  
    The permutation matrix corresponding to \( f \) is called a CPM.  
    A QC-LDPC code is a classical linear code whose parity-check matrix \( H \) is constructed by arranging \( J \times L \) CPMs.  
    The matrix \( H \) can also be defined using a function array \( F(H) \), where each CPM corresponds to a permutation mapping.  
    Specifically, \( F_{jl} \sim f_{jl} \), and the following representations hold:  
    \begin{align}
        H =
        \begin{bmatrix}
            F_{11} & \cdots & F_{1L} \\
            \vdots & \ddots & \vdots \\
            F_{J1} & \cdots & F_{JL}
        \end{bmatrix}
        , \quad
        F(H) =
        \begin{bmatrix}
            f_{11} & \cdots & f_{1L} \\
            \vdots & \ddots & \vdots \\
            f_{J1} & \cdots & f_{JL}
        \end{bmatrix}.
    \end{align}
\end{dfn}

In general, QC-LDPC matrices are represented by an array of exponents consisting only of \( b \) in the function \( f = x + b \).
However, in this paper, we represent QC-LDPC matrices using the function array \( F(H) \) as shown above.
Furthermore, we will use 0-based indexing for the matrices in this paper.

Hereafter, unless otherwise stated, we assume that \( F \sim f \) and \( G \sim g \) represent CPMs and their corresponding CPM mappings of size \( P \). 
Additionally, while the composition of \( f \) and \( g \) is usually denoted as \( f \circ g \) or \( g \circ f \), for simplicity, we write them as \( fg \) and \( gf \), respectively. 
The product of two CPMs, \( FG \), corresponds to the composition of their respective mappings, \( fg \).

\section{{Cycles in the LDPC Matrix}}\label{sec:Cycles in the LDPC Matrix}
The performance of the sum-product (SP) algorithm significantly deteriorates in the presence of short cycles in the Tanner graph \cite{mackay1999good}. 
To mitigate this issue, it is desirable to construct a parity-check matrix that maximizes the minimum cycle length in the Tanner graph, referred to as the girth. 
In this paper, we simply state that a cycle of length \( n \) exists in \( H \) 
if the Tanner graph defined by the parity-check matrix \( H \) of a classical code \( C \) contains such a cycle.

Since a permutation matrix has a row weight and column weight of 1, the 1-symbols corresponding to adjacent nodes in the Tanner graph appear in different permutation matrix blocks within the same row or column. 
Thus, a cycle in an LDPC matrix can be represented as a sequence of permutation matrix blocks and, equivalently, as a sequence of the corresponding permutation mappings.
Moreover, this sequence can be expressed as a composition of the corresponding permutation mappings, providing an algebraic characterization of cycles in LDPC matrices.

It is known that the girth of a general QC-LDPC matrix is upper-bounded by 12 \cite{fossorier2004quasicyclic}.  

\section{{Orthogonal QC-LDPC Matrix Pairs with Girth 12 for Quantum Error Correction}}
\label{sec:Orthogonal QC-LDPC Matrix Pairs with Girth 12 for Quantum Error Correction}

In this section, we first present the construction method for classical QC-LDPC matrices in Subsection~\ref{sec:Explicit Construction of Classical QC-LDPC Matrices with Girth 12},
which is necessary for constructing quantum QC-LDPC codes, and demonstrate that this construction achieves a girth of 12.
Next, in Subsection~\ref{sec:quantum-QC-LDPC-matrices},
we construct orthogonal QC-LDPC matrix pairs for quantum error correction based on this classical construction.

\subsection{{Explicit Construction of Classical Codes}}
\label{sec:Explicit Construction of Classical QC-LDPC Matrices with Girth 12}

{
While explicit constructions of QC-LDPC matrices with girth at least 6 \cite{lally2007explicit}, 
8 \cite{tasdighi2016symmetrical,zhang2014deterministic,zhang2024explicit}, 
and 10 \cite{wang2014explicit} have been reported in the literature, 
achieving the maximum girth of 12 has typically relied on randomly generating and searching candidate matrices.
}

{
Hereafter, let \( J = 2 \). This is due to three reasons: first, the girth upper bound for quantum QC-LDPC codes with \( J \ge 3 \) is 6; second, the performance of NB-LDPC codes is highest when \( J = 2 \); and third, the finite field extension method for orthogonal QC-LDPC matrix pairs from \cite{kasai2011quantum} targets the case where \( J = 2 \).
In this subsection, we present an explicit construction method for a \( (2, L, P) \)-QC-LDPC matrix that achieves the maximum possible girth of 12.  
By reaching this upper bound, we completely eliminate the need for a random search process. 
}


This construction method can be seen as a special case of the classical QC-LDPC matrix construction 
based on Sidon sequences \cite[Construction I]{zhang2022construction}. 
However, unlike \cite{zhang2022construction}, 
our method provides an explicit and concise specification of the parameters for each block corresponding to the Sidon sequence.

\begin{dfn}
    [Construction of Classical QC-LDPC Matrices]  
    \label{dfn:classicalQC-LDPC-matrix}  
    Let \( L \geq 6 \) be an even integer and \( P \) be a positive integer.  
    Define the set of CPM mappings \( \{f_0,\cdots,f_{L/2-1}\} \in \mathcal{F}_P^{L/2} \) as follows:  
    \begin{align}
        f_l(x) \defeq x + 2^l, \
        g_l(x) \defeq x + 2^{l+L/2}
        \quad \text{for } l \in [L/2]. 
    \end{align}
    Based on this, we construct the \( (2,L/2,P) \)-QC-LDPC matrices \( H, H' \) by defining the function arrays \( F(H), F(H') \) as follows:  
    \begin{align}
        \bigl(F(H)\bigr)_{jl} \defeq f_{-j+l}, 
        \quad
        \bigl(F(H')\bigr)_{jl} \defeq g_{-j+l}
    \end{align}
    for \( j \in \{0,1\} \).  
    Here, the indices of \( f \) are taken modulo \( L/2 \).  
    The resulting \( (2,L,P) \)-QC-LDPC matrix is constructed as \( \HH = [H,H'] \).
\end{dfn}



\begin{thr}
    \label{thr:2section-classicalQC-QLDPC-matrix-achieve-girth12}  
    For any \( L \geq 6 \) and \( P \geq P_0 \defeq 2^{L+1} \), 
    the \( (2, L, P) \)-QC-LDPC matrix \( \HH \) constructed according to Definition \ref{dfn:classicalQC-LDPC-matrix} has a girth of 12.
\end{thr}





Note that the following is true:

\begin{prp}
    \label{prp:L=4-exists-8cycle}
    In the case of $L=4$, the QC-LDPC matrix $\HH$ from Theorem \ref{thr:2section-classicalQC-QLDPC-matrix-achieve-girth12} contains an 8-cycle. 
\end{prp}


The construction method in Definition~\ref{dfn:classicalQC-LDPC-matrix} corresponds to the structure obtained in \cite[Construction~I]{zhang2022construction} 
when the number of circ-blocks is fixed to 1 and the Sidon sequence is chosen as
$S = \{2^0, 2^1, \dots, 2^{L-1}\}$.


\subsection{{Explicit Construction of Orthogonal QC-LDPC Matrix Pairs}}
\label{sec:quantum-QC-LDPC-matrices}
Previously, explicit constructions for quantum QC-LDPC codes have been limited 
to avoiding cycles of length up to 4 \cite{hagiwara2011spatially}.

In this subsection, we construct an orthogonal LDPC matrix pair \((\HH_X, \HH_Z)\) for quantum QC-LDPC codes. 
Building upon the QC-LDPC matrices constructed in Subsection \ref{sec:Explicit Construction of Classical QC-LDPC Matrices with Girth 12}, 
we now present a method to obtain an orthogonal QC-LDPC matrix pair that achieves girth 12.

\begin{dfn}
    [Construction of Orthogonal QC-LDPC Matrix Pairs]
    \label{dfn:orthogonal-QC-LDPC-matrices}
    Let \( L \geq 6 \) be a positive even integer, and let \( P \) be a positive integer. 
    Define the sets of CPM mappings \( \{ f_0, \dots, f_{L/2-1} \}, \{ g_0, \dots, g_{L/2-1} \} \in \mathcal{F}_P^{L/2} \) as follows:
    \begin{align}
        f_l(x) &\defeq x + 2^l, \quad
        g_l(x) \defeq x + 2^{l + L/2} \quad \text{for } l \in [L/2].
    \end{align}
    Using these mappings, we construct the function arrays \( F(\hat{H}_X) \) and \( F(\hat{H}_Z) \) that define the QC-LDPC matrix pair \( (\hat{H}_X, \hat{H}_Z) \) as follows:
    \begin{align}
        \bigl(F(\hat{H}_X)\bigr)_{jl} &\defeq 
        \begin{cases}
            f_{-j+l}, & 0 \leq l < L/2, \\
            g_{-j + l - L/2}, & L/2 \leq l < L, 
        \end{cases}
        \\
        \bigl(F(\hat{H}_Z)\bigr)_{jl} &\defeq
        \begin{cases}
            g_{j - l}^{-1}, & 0 \leq l < L/2, \\
            f_{j - l + L/2}^{-1}, & L/2 \leq l < L, 
        \end{cases}
        \label{dfn_qunatum_protograph_codes_construction}
    \end{align}
    for \( j \in \{0,1\} \).
    The indices of \( f \) and \( g \) are taken modulo \( L/2 \).
\end{dfn}

\( F(\HH_X) \) and \( F(\HH_Z) \) construct the first row (\( j=1 \)) by shifting the block indices of the zeroth row (\( j=0 \)).  
This arrangement ensures the orthogonality of \( \HH_X \) and \( \HH_Z \) by employing the construction method for orthogonal QC-LDPC matrix pairs proposed by Hagiwara and Imai \cite{hagiwara2007quantum}.  
The construction method in Definition~\ref{dfn:orthogonal-QC-LDPC-matrices} is equivalent to the construction method in \cite{hagiwara2007quantum} with \( \sigma = 2 \) and \( \tau = 1 \). 
Therefore, the orthogonality of the QC-LDPC matrices \( \HH_X \) and \( \HH_Z \) constructed in Definition \ref{dfn:orthogonal-QC-LDPC-matrices} follows from \cite[Proposition 4.1]{hagiwara2007quantum}.

Next, we analyze the cycles in this QC-LDPC matrix pair.  
The matrices \( \HH_X \) and \( \HH_Z \) are graph isomorphic, meaning that they can be transformed into each other solely by row and column permutations.  
As a result, the cycles in \( \HH_X \) and \( \HH_Z \) correspond one-to-one, with each having cycles of the same length. 
Hagiwara and Imai provided a proof for the case \( J = L/2 \) in \cite[Proposition 4.2]{hagiwara2007quantum}.  
We extend this result to cases where \( J \neq L/2 \).

\begin{thr}
    [Graph Isomorphism of Orthogonal QC-LDPC Matrix Pairs, Generalization of {\cite[Proposition 4.2]{hagiwara2007quantum}}]
    \label{thr:H_X-H_Z-equiv}
    The orthogonal \((J,L,P)\)-QC-LDPC matrix pair \((\hat{H}_X,\hat{H}_Z)\) constructed using Definition \ref{dfn:orthogonal-QC-LDPC-matrices} is graph isomorphic.  
    In other words, \( \hat{H}_Z \) can be obtained from \( \hat{H}_X \), and vice versa, by applying the same row and column permutation operations.
\end{thr}


Next, we present the theorem stating that \( \HH_X \) achieves a girth of 12.  
From Theorem \ref{thr:H_X-H_Z-equiv}, it follows that \( \HH_Z \) has the same girth as \( \HH_X \).

\begin{thr}
    \label{thr:QC-QLDPC-achieve-girth12}
    For any \( L \geq 6 \) and \( P \geq P_0 \defeq 2^{L+1} \), the QC-LDPC matrix \( \HH_X \) constructed according to Definition \ref{dfn:orthogonal-QC-LDPC-matrices} has a girth of 12.
\end{thr}


Theorem \ref{thr:QC-QLDPC-achieve-girth12} holds for any \( P \ge P_0 = 2^{L+1} \), but it is known that there exist values of \( P < P_0 \) for which this theorem remains valid. In particular, when constructing codes with short block lengths, the minimum value of \( P \) that achieves girth 12 for each \( L \) is of significant importance. Table \ref{tab:minimumPforL-girth12} presents these minimum values. It is observed that the required minimum circulant size \( P \) for obtaining QC-LDPC matrix pairs that achieve girth 12 (as guaranteed by Theorem \ref{thr:QC-QLDPC-achieve-girth12}) is actually much smaller than \( P_0 \).

\begin{table}[htbp]
    \centering
    \caption{For each \( L \geq 6 \), the minimum value of \( P \), denoted as \( P_{\text{min}} \), for which Theorem \ref{thr:QC-QLDPC-achieve-girth12} holds with \( P < P_0 = 2^{L+1} \).}
    \begin{tabular}{c|cccccc}
        \( L \) & 6  & 8  & 10  & 12  & 14  & 16  \\ \hline\hline
        \( P_0 \) & 128 & 512 & 2048 & 8192 & 32768 & 131072 \\ \hline
        \( P_{\text{min}} \) & 49  & 138  & 281  & 355  & 609  & 821  \\ 
    \end{tabular}
    \vspace{5pt} 
    \label{tab:minimumPforL-girth12}
\end{table}

As an example, Figure \ref{fig:orthogonal-QC-LDPC-matrices-L6P49-label} provides an example of \( F(\HH_X) \) and \( F(\HH_Z) \) that define the orthogonal QC-LDPC matrix pair \( (\HH_X, \HH_Z) \) for \( L = 6 \), \( P = 49 \).

\begin{figure*}[htbp]
    \centering
    \begin{alignat}{3}
        F(\HH_X)
        &=
        \left[
        \begin{array}{ccc|ccc}
            x + 2^0 & x + 2^1 & x + 2^2 & x + 2^3 & x + 2^4 & x + 2^5 \\
            x + 2^2 & x + 2^0 & x + 2^1 & x + 2^5 & x + 2^3 & x + 2^4
        \end{array}
        \right]
        &=&
        \begin{bmatrix}
            x + \ 1 \, & x + \ 2 \, & x + \ 4 & x + \ 8 & x + 16 & x + 32 \\
            x + \ 4 \, & x + \ 1 \, & x + \ 2 & x + 32 & x + \ 8 & x + 16
        \end{bmatrix}
        , \\
        F(\HH_Z)
        &=
        \left[
        \begin{array}{ccc|ccc}
            x - 2^3 & x - 2^5 & x - 2^4 & x - 2^0 & x - 2^2 & x - 2^1 \\
            x - 2^4 & x - 2^3 & x - 2^5 & x - 2^1 & x - 2^0 & x - 2^2
        \end{array}
        \right]
        &=&
        \begin{bmatrix}
            x + 41 & x + 17 & x + 33 & x + 48 & x + 45 & x + 47 \\
            x + 33 & x + 41 & x + 17 & x + 47 & x + 48 & x + 45
        \end{bmatrix}&.
    \end{alignat}
    \caption{
    Example of \( F(\HH_X) \) and \( F(\HH_Z) \) for \( L = 6 \), \( P = 49 \).
    }
    \label{fig:orthogonal-QC-LDPC-matrices-L6P49-label}
\end{figure*}


\section{{Finite Field Extension}}
\label{sec:finite-field-extension}
In this section, we extend the orthogonal matrix pair \( (\HH_X,\HH_Z) \in \Fb_2^{PJ \times PL} \) 
constructed in Section \ref{sec:Orthogonal QC-LDPC Matrix Pairs with Girth 12 for Quantum Error Correction} 
to the orthogonal matrix pair \( (H_\Gamma, H_\Delta) \in \Fb_{q=2^e}^{PJ \times PL} \) using the finite field extension method by Kasai et al. \cite{kasai2011quantum}. 
This allows us to explicitly construct an NB-QC-LDPC matrix pair that preserves both orthogonality and the maximum girth of 12.
In general, it is known that the performance of NB-LDPC codes is maximized when \( J=2 \) \cite{poulliat2008design}. 
Therefore, quantum LDPC codes constructed from such orthogonal QC-LDPC matrix pairs with \( J=2 \) and girth 12 are expected to exhibit high decoding performance.


Owing to page limitations, we include only the essential theorems required to support our claims.
For the complete finite field extension method, please refer to \cite{komoto2024quantum}.

The matrix pair \( (H_\Gamma = (\gamma_{jl}),H_\Delta = (\delta_{jl})) \) is constructed by replacing the 1-symbols in \( (\HH_X,\HH_Z) \) with nonzero elements of \( \Fb_{q} \). 
Komoto et al. \cite{komoto2024quantum} provided the conditions that the sets \( \{f_0,\cdots,f_{L/2-1}\}, \{g_0,\cdots,g_{L/2-1}\} \in \mathcal{F}_P^{L/2} \) must satisfy to extend a broader class of matrix pairs over finite fields using the algorithm of \cite{kasai2011quantum} than the QC-LDPC matrix pair in \cite{hagiwara2007quantum}.

\begin{thr}
    [Conditions for Finite Field Extension, \cite{komoto2024quantum}]
    \label{thr:condtion-finite-field-extension}
    If the sets \( \{f_0,\cdots,f_{L/2-1}\}, \{g_0,\cdots,g_{L/2-1}\} \in \mathcal{F}_P^{L/2} \) satisfy the following condition, an orthogonal LDPC matrix pair \( (H_\Gamma, H_\Delta) \in \Fb_{q=2^e}^{PJ \times PL} \) can be obtained by extending the 1-symbols of the orthogonal QC-LDPC matrix pair \( (\HH_X,\HH_Z) \in \Fb_2^{PJ \times PL} \), defined in Definition \ref{dfn:orthogonal-QC-LDPC-matrices}, over a finite field:
    \begin{align}
        f_{l-k} g_{-l+k'}(x) \neq f_{l'-k} g_{-l'+k'}(x), 
    \end{align}
    for $l \neq l' \in [L/2], k,k' \in \{0,1\}$. 
    Here, the indices of \( f,g \) are taken modulo \( L/2 \).
\end{thr}

We demonstrate that the QC-LDPC matrix pair \( (\HH_X, \HH_Z) \) 
constructed in Definition \ref{dfn:orthogonal-QC-LDPC-matrices} satisfies 
Theorem \ref{thr:condtion-finite-field-extension}, 
meaning that it is extendable to a finite field.

\begin{thr}
    \label{thr:q_nonbiqc_withgirth_12}
    [Finite Field Extension of QC-LDPC Matrix Pair]
    The QC-LDPC matrix pair constructed in Definition \ref{dfn:orthogonal-QC-LDPC-matrices} satisfies the conditions of Theorem \ref{thr:condtion-finite-field-extension}.
\end{thr}


From Theorem \ref{thr:q_nonbiqc_withgirth_12}, it follows that the QC-LDPC matrix pair constructed by Definition \ref{dfn:orthogonal-QC-LDPC-matrices} can be extended to a non-binary form using the method proposed in \cite{kasai2011quantum}.

\section{{Numerical Results}}

We conducted numerical experiments using joint belief propagation decoding over depolarizing channels with error probability~$p$, under the parameters $e = 8$, $J = 2$, $L \in \{8, 10, 12\}$, and $P \in \{32, 128, 1024, 8129\}$. Degenerate errors were not taken into account; decoding was regarded as successful if the noise pattern was correctly estimated, and as a failure otherwise.

The results closely match those presented in Figure~1 of~\cite{komoto2024quantum}. 
Performance plots are omitted due to space constraints, 
but can be provided as supplementary material upon request.

\section{{Spatial Coupling}}
\label{sec:spatial-coupling}
In this section, we apply the orthogonal NB-QC-LDPC matrix pairs constructed 
in Section \ref{sec:finite-field-extension} 
to construct spatially-coupled variants that also achieve girth 12.

Hagiwara et al. \cite{hagiwara2011spatially} proposed a construction method for spatially-coupled QC-LDPC matrix pairs \((H_X, H_Z)\) using QC-LDPC matrix pairs \((\HH_X^{i_c}, \HH_Z^{i_c})\) with coupling number \(n_c\) for \(i_c \in [n_c]\). 
When using the method in \cite{hagiwara2011spatially}, 
cycles of length 8 may form between adjacent sections.
In this study, we address this issue by introducing slight variations across sections, thereby maintaining a girth of 12.

\begin{dfn}
    [Construction of Spatially-Coupled QC-LDPC Matrix Pairs]
    \label{dfn:SC-QC-LDPC-matrices}
    We follow the notation and construction in \cite[Fig.1]{hagiwara2011spatially}. 
    Let \(n_c\) be a positive integer representing the spatial coupling number. 
    We define \(n_c\) NB-LDPC matrix pairs \((H_X^{i_c},H_Z^{i_c})\) as NB-QC-LDPC matrix pairs constructed from the function arrays \(F(\HH_X^{i_c}), F(\HH_Z^{i_c})\) defined below.
    For even \(i_c\),
    \begin{align}
        f_l(x) = x + 2^l
        , \quad
        g_l(x) = x + 2^{l + L/2}
        \quad \text{for } l \in [L/2]. 
    \end{align}
    For odd \(i_c\),
    \begin{align}
        f_l(x) = x + 2^{l + L}
        , \quad
        g_l(x) = x + 2^{l + 3L/2}
        \quad \text{for } l \in [L/2]. 
    \end{align}
    Then, by arranging \(H_X^{i_c},H_Z^{i_c}\) as follows, we construct the expanded parity-check matrix pair \((H_X,H_Z)\) of size \((n_c+1)P/2 \times n_cPL\). 
    Adjacent sections \(H_{X/Z}^{i_c}, H_{X/Z}^{i_c+1}\) for \(i_c \in [n_c-1]\) are arranged to share \(P\) rows. 
    \begin{align}
        H_X &=
        \scriptsize{
        \begin{bmatrix}
            \multirow{2}{*}{\normalsize $H_{X}^0$} &  &  &  \\
             & \multirow{2}{*}{\normalsize $H_{X}^1$} &  &  \\
             &  & \multirow{2}{*}{\normalsize $\dddots$} &  \\
             &  &  & \multirow{2}{*}{\normalsize $H_{X}^{n_c-1}$} \\
             &  &  & 
        \end{bmatrix}}, \\
        H_Z &=
        \scriptsize{
        \begin{bmatrix}
             &  &  & \multirow{2}{*}{\normalsize $H_{Z}^{n_c-1}$} \\
             &  & \multirow{2}{*}{\normalsize $\iiddots$} &  \\
             & \multirow{2}{*}{\normalsize $H_{Z}^1$} &  &  \\
            \multirow{2}{*}{\normalsize $H_{Z}^0$} &  &  &  \\
             &  &  & 
        \end{bmatrix}
        }.
    \end{align}
\end{dfn}

It is clear that the SCNB-QC-LDPC matrix pair \((H_X, H_Z)\) constructed according to Definition \ref{dfn:SC-QC-LDPC-matrices} is orthogonal.

We will show that \((H_X, H_Z)\) maintains girth 12. 
As in the case of non-spatial coupling, \(H_X\) and \(H_Z\) are graph isomorphic. 
Therefore, the girth values of \(H_X\) and \(H_Z\) are identical.

\begin{thr}
    [Graph Isomorphism of Orthogonal SC-QC-LDPC Matrix Pair]
    \label{thr:SC-H_X-H_Z-equiv}
    For the SCNB-QC-LDPC matrix pair \((H_X, H_Z)\) constructed by Definition \ref{dfn:SC-QC-LDPC-matrices}, 
    \(H_X\) and \(H_Z\) are graph isomorphic.
\end{thr}

The construction in Definition \ref{dfn:SC-QC-LDPC-matrices} utilizes the parity-check matrix pair from Definition \ref{dfn:orthogonal-QC-LDPC-matrices}, ensuring that no cycles of size 8 or smaller exist within each section. 
Additionally, it is clear that cycles of size 4 spanning two sections, cycles of size 8 or smaller spanning non-adjacent sections, and cycles of size 8 or smaller spanning three sections do not exist in \(H_X\). Therefore, to prove that \((H_X, H_Z)\) achieves a girth of 12, it suffices to show that no size 8 cycles span adjacent sections in \(H_X\).

\begin{thr}
    \label{thr:spatial_coupling_girth12}
    For any \(P \ge P_1 \defeq 2^L(2^L+1)\), the SC-QC-LDPC matrix \(H_X\) constructed according to Definition \ref{dfn:SC-QC-LDPC-matrices} achieves a girth of 12.
\end{thr}

Based on the above, it has been demonstrated that a quantum LDPC code can be constructed using a classical LDPC matrix pair with the SCNB-QC-LDPC matrix pair, which has a girth of 12, the upper bound for QC-LDPC matrices.

Up to this point, the construction has been done using a base of 2 for the exponents corresponding to each permutation matrix. It has been shown that if the base of the exponent is extended to \( n > 2 \), the same reasoning as for \( n = 2 \) can still achieve a girth of 12, and this can be naturally extended for \( n > 1 \). 
However, the values are \( P_0 = 2n^{L-1} \) and \( P_1 = 2n^{L-1}(n^L + 1) \).

\section{{Conclusion}}

In this paper, we proposed a deterministic method for the explicit construction of orthogonal QC-LDPC matrix pairs for quantum LDPC codes that achieve a girth of 12.
Unlike conventional approaches that rely on random search to obtain such codes, 
our proposed method deterministically constructs the required parity-check matrices.
Furthermore, the proposed construction method can be extended to non-binary and spatially-coupled LDPC codes 
while preserving the girth of 12.
These results contribute to the efficient design of structured quantum error-correcting codes 
and offer a foundation for the realization of fault-tolerant quantum communication and scalable quantum computation.

\bibliographystyle{IEEEtran}
\bibliography{main}

\begin{appendices}
    \input{appendix}
\end{appendices}

\end{document}

%% file: appendix.tex
\section{Proofs in Section \ref{sec:Explicit Construction of Classical QC-LDPC Matrices with Girth 12}}
\begin{dfn}
    [Composition of Permutation Mappings Corresponding to a Block Sequence \( f^* \)]  
    Given an LDPC matrix \( H \) defined by the function chain \( F(H) = (f_{j,l}) \), consider the following sequence of indices:  
    \begin{align}
        [(j_0,l_0), (j_0,l_1), (j_1,l_1), 
        \cdots
        , (j_{n-1},l_n), (j_n,l_n)].
        \label{align_block_route}
    \end{align}
    This sequence represents a path in the function chain where transitions occur within the same row or column.  
    Define \( f_{j_i,l_i} = f_i \) and \( f_{j_i,l_{i+1}} = g_i \).  
    Then, the composition of permutation mappings corresponding to the block sequence, denoted as \( f^* \), is defined as:  
    \begin{align}
        \label{dfn_f*}
        f^*(x) \defeq
        \big(g_n^{-1} f_n g_{n-1}^{-1} f_{n-1} \cdots g_1^{-1} f_1\big) (x).
    \end{align}
\end{dfn}

Cycles present in the parity-check matrix can be represented using the composition of permutation mappings \( f^* \) defined above \cite[Theorem 1]{myung2006combining}.

\begin{thr}
    [Cycle Representation Using the Composition of Permutation Mappings, {\cite[Theorem 1]{myung2006combining}}]  
    \label{thr:2nCycle_myung2006combinig_thr1}  
    Let \( n \) be the smallest integer such that \( (j_0, l_0) = (j_n, l_n) \) in \eqref{align_block_route}.  
    Then, the following two statements are equivalent:  
    \begin{enumerate}
        \item There exists \( x \in \Zb_P \) such that \( f^*(x) = x \).  
        \item The parity-check matrix \( H \) contains a cycle of length \( 2n \) corresponding to the block sequence in \eqref{align_block_route}.  
    \end{enumerate}
\end{thr}

Theorem \ref{thr:2nCycle_myung2006combinig_thr1} is useful for analyzing the existence of cycles in the parity-check matrix. 

\begin{dfn}
    \label{dfn:classicalQC-LDPC-matrix-half}  
    Let \( L \geq 6 \) be a positive even integer and \( P \) be a positive integer.  
    Define a set of CPM mappings \( \{ f_0, \dots, f_{L/2-1} \} \in \mathcal{F}_P^{L/2} \) as follows:  
    \begin{align}
        f_l(x) \defeq x + 2^l
        \quad \text{for } l \in [L/2].
    \end{align}
    Then, we construct a \( (2, L/2, P) \)-QC-LDPC matrix \( H \) by defining the function array \( F(H) \) as follows:  
    \begin{align}
        \bigl(F(H)\bigr)_{jl} \defeq f_{-j+l} \quad \text{for } j \in \{0,1\}. 
    \end{align}
    The indices of \( f \) are considered over \( \mathbb{Z}_{L/2} \).  
\end{dfn}

\begin{thr}
    \label{thr:classicalQC-LDPC-matrix-achieve-girth12-half}  
    For any \( L \geq 6 \) and \( P \geq P_0 \defeq 2^{L/2+1} \), the \( (2, L/2, P) \)-QC-LDPC matrix \( H \) constructed according to Definition \ref{dfn:classicalQC-LDPC-matrix} has a girth of 12.  
\end{thr}

\begin{proof}
    [Proof of Theorem \ref{thr:classicalQC-LDPC-matrix-achieve-girth12-half}]
    The theorem follows from Lemma \ref{lmm:classicalQC-LDPC-matrix-no4cycle} and Lemma \ref{lmm:classicalQC-LDPC-matrix-no8cycle}.  
\end{proof}

\begin{lmm}  
    \label{lmm:classicalQC-LDPC-matrix-no4cycle}  
    For any \( L \geq 6 \) and \( P \geq P_0' = 2^{L/2} \), the QC-LDPC matrix \( H \) constructed according to Definition \ref{dfn:classicalQC-LDPC-matrix} does not contain any cycles of length 4.  
\end{lmm}

\begin{proof}
    [Proof of Lemma \ref{lmm:classicalQC-LDPC-matrix-no4cycle}]
    Consider the path corresponding to four circulant permutation matrix (CPM) blocks in \( H \), which corresponds to the transformation \( f^*(x) \defeq x + b^* \). The value \( b^* \) can be expressed as follows:  
    \begin{align}
        b^*
        &= b_{l} - b_{l'} + b_{l'-1} - b_{l-1}
        \\&= (b_{l} + b_{l'-1}) - (b_{l-1} + b_{l'}).
    \end{align}
    Here,  
    \begin{align}
        (b_{l}, b_{l-1}) =
        \begin{cases}
            (2^0, 2^{L/2-1}), & l = 0 \\
            (2^{l}, 2^{l-1}), & \text{otherwise}.
        \end{cases}
    \end{align}
    Additionally, \( l \neq l' \).  
    We need to show that \( b^* \neq 0 \) and \( |b^*| < P_0' \).  
    \begin{align}
        (b_{l} + b_{l'-1}, b_{l-1} + b_{l'}) = (2^l + 2^{l'-1}, 2^{l-1} + 2^{l'}).
    \end{align}
    Considering the binary representation, since \( b_{l} + b_{l'} \neq b_{l-1} + b_{l'-1} \), it follows that \( b^* \neq 0 \).  
    Furthermore, we obtain the following bound:  
    \begin{align}
        |b^*|
        &\leq |2^0 - 2^{L/2-1} + 2^{L/2-2} - 2^{L/2-1}|
        \\&= |2^0 + 2^{L/2-2} - 2^{L/2}|
        \\&< 2^{L/2} = P_0'.
    \end{align}
\end{proof}

\begin{lmm}  
    \label{lmm:classicalQC-LDPC-matrix-no8cycle}  
    For any \( L \geq 6 \) and \( P \geq P_0 = 2^{L/2+1} \), the QC-LDPC matrix \( H \) constructed according to Definition \ref{dfn:classicalQC-LDPC-matrix} does not contain any cycles of length 8.  
\end{lmm}

\begin{proof}
    [Proof of Lemma \ref{lmm:classicalQC-LDPC-matrix-no8cycle}]
    The function \( f^*(x) \) corresponding to the path through the eight circulant permutation matrix (CPM) blocks in \( \HH_X \) is given by:  
    \begin{align}  
        b^*  
        =& b_{l_0} - b_{l_1} + b_{l_1-1} - b_{l_2-1} + b_{l_2} - b_{l_3} + b_{l_3-1} - b_{l_0-1}  
        \\=& (b_{l_0} + b_{l_1-1} + b_{l_2} + b_{l_3-1})  
        \\& - (b_{l_0-1} + b_{l_1} + b_{l_2-1} + b_{l_3}).  
    \end{align}
    Here,  
    \begin{align}  
        (b_{l_i}, b_{l_i-1}) =  
        \begin{cases}  
            (2^0, 2^{L/2-1}), & l_i = 0 \\  
            (2^{l_i}, 2^{l_i-1}), & \text{otherwise},  
        \end{cases}  
    \end{align}  
    for \( i = 0, \dots, 3 \). 
    Furthermore, \( l_i \neq l_{i+1} \) for \( i \in \mathbb{Z}_4 \).  
    To prove the lemma, we need to show that \( b^* \neq 0 \) and \( |b^*| < 2^{L/2+1} \).  
    These properties can be established using the same argument as in Lemma \ref{lmm:classicalQC-LDPC-matrix-no4cycle}.  
\end{proof}

\begin{proof}
    [Proof of Theorem \ref{thr:2section-classicalQC-QLDPC-matrix-achieve-girth12}]
    The proof follows directly from Theorem \ref{thr:classicalQC-LDPC-matrix-achieve-girth12-half}. Since both \(H\) and \(H'\) are constructed using the method described in Definition \ref{dfn:classicalQC-LDPC-matrix}, they individually achieve a girth of 12. 

    Moreover, when concatenated to form \(\HH = [H, H']\), the structure ensures that no additional short cycles (of length less than 12) are introduced. The girth of the resulting matrix remains 12, satisfying the theorem’s claim.
\end{proof}

\begin{proof}
    [Proof of Proposition \ref{prp:L=4-exists-8cycle}]
    Consider the index sequence of the block path as follows:
    \begin{align}
        [(0,0),(0,2),(1,2),(1,1),(0,1),(0,3),(1,3),(1,0),(0,0)].
    \end{align}
    The corresponding $f^*$ for this index sequence is given as:
    \begin{align}
        f^*(x)
        =& x + b_{0,0} - b_{0,2} + b_{1,2} - b_{1,1}
        \\& + b_{0,1} - b_{0,3} + b_{1,3} - b_{1,0}
        \\=& x + 2^{0} - 2^{2} + 2^{3} - 2^{0} + 2^{1} - 2^{3} + 2^{2} - 2^{1}
        \\=& x.
    \end{align}
\end{proof}

\section{Proofs in Section \ref{sec:Orthogonal QC-LDPC Matrix Pairs with Girth 12 for Quantum Error Correction}}
\begin{proof}
    [Proof of Theorem \ref{thr:H_X-H_Z-equiv}]
    Define \( P_R \) and \( P_C \) as follows:
    \begin{align}
        P_R \defeq Q, \quad
        P_C \defeq \left[
        \begin{array}{cc}
            O & R  \\
            R & O
        \end{array}
        \right].
    \end{align}
    Here, \( Q \in \mathcal{M}_{PJ} \sim - x + PJ - 1 \in \mathcal{F}_{PJ} \), and \( R \in \mathcal{M}_{PL/2} \sim - x + PJ-1 \in \mathcal{F}_{PL/2} \) are defined accordingly. 
    Then, the following holds:
    \begin{align}
        \hat{H}_Z = P_R \hat{H}_X P_C, \quad
        \hat{H}_X = P_R \hat{H}_Z P_C.
    \end{align}
\end{proof}

\begin{proof}
    [Proof of Theorem \ref{thr:QC-QLDPC-achieve-girth12}]
    The parity-check matrix \( \HH_X \) is identical to the matrix \( \HH \) in Theorem \ref{thr:2section-classicalQC-QLDPC-matrix-achieve-girth12}. Therefore, the result follows directly from Theorem \ref{thr:2section-classicalQC-QLDPC-matrix-achieve-girth12}.
\end{proof}

\section{Proofs in Section \ref{sec:finite-field-extension}}
\begin{proof}
    [Proof of Theorem \ref{thr:q_nonbiqc_withgirth_12}]
    Expanding $f_l g_{-l+k}$ for each $k=0,\pm 1$ with $l \in [L/2]$, we obtain the following expressions:
    \begin{align}
        f_l g_{-l}(x) &= x + 2^{l} + 2^{L-l}, \\
        f_l g_{-l-1}(x) &= x + 2^{l} + 2^{L-l-1}, \\
        f_l g_{-l+1}(x) &= x + 2^{l} + 2^{L-l+1}.
    \end{align}
    Here, the exponents of the last terms in each expression, namely $2^{L-l},2^{L-l-1},2^{L-l+1}$, are defined on $L/2+\Zb_{L/2}$. 
    Thus, the conditions of Theorem \ref{thr:condtion-finite-field-extension} are satisfied.
\end{proof}

\section{Proofs in Section \ref{sec:spatial-coupling}}
\begin{proof}
    [Proof of Theorem \ref{thr:SC-H_X-H_Z-equiv}]
    Define \( P'_R \) and \( P'_C \) as follows. 
    \( P'_C \) has \( P_C \) repeated \( n_c \) times along the diagonal.
    \begin{align}
        P'_R &\defeq Q', \quad
        P'_C \defeq
        \begin{bmatrix}
            P_C & & \\
             & \ddots & \\
             & & P_C
        \end{bmatrix}.
    \end{align}
    Here, \( Q' \in \Mc_{(n_c+1)P} \sim - x + (n_c+1)P - 1 \in \Fc_{(n_c+1)P} \), and \( P_C \) is defined in the proof of Theorem \ref{thr:H_X-H_Z-equiv}.
    Then, the following calculations hold.
    \begin{align}
        H_Z &= P'_R H_X P'_C, \quad
        H_X = P'_R H_Z P'_C.
    \end{align}
\end{proof}

\begin{proof}
    [Proof of Theorem \ref{thr:spatial_coupling_girth12}]
    This is proven by Theorem \ref{thr:QC-QLDPC-achieve-girth12} and Lemma \ref{lmm:adjacent-section-no8cycle} provided below.
\end{proof}

\begin{lmm}
    \label{lmm:adjacent-section-no8cycle}
    The LDPC matrix \( H_X \) constructed by Definition \ref{dfn:SC-QC-LDPC-matrices} does not contain 8-cycles between adjacent sections.
\end{lmm}

\begin{proof}
    [Proof of Lemma \ref{lmm:adjacent-section-no8cycle}]
    It suffices to show that there is no 8-cycle between the sections \( i_c = 0 \) and \( i_c = 1 \). 
    The path corresponding to the 8 permutation matrix blocks spanning these two sections can be expressed as follows, with respect to \( f^*(x) = x + b^* \):
    \begin{align}
        b^*
        =& b_{l_0} - b_{l_1} + b_{l_1-1} - c_{l_2-1} + c_{l_2} - c_{l_3} + c_{l_3-1} - b_{l_0-1} \\
        =& (b_{l_0} + b_{l_1-1} + c_{l_2} + c_{l_3-1}) \\
        &- (b_{l_0-1} + b_{l_1} + c_{l_2-1} + c_{l_3}).
    \end{align}
    Here,
    \begin{align}
        (b_{l_i}, b_{l_i-1}) &=
        \begin{cases}
            (2^0, 2^{L/2-1}), & l_i = 0 \\
            (2^{L/2}, 2^{L-1}), & l_i = L/2 \\
            (2^{l_i}, 2^{l_i-1}), & \text{otherwise},
        \end{cases} \\
        (c_{l_i}, c_{l_i-1}) &=
        \begin{cases}
            (2^L, 2^{3L/2-1}), & l_i = 0 \\
            (2^{3L/2}, 2^{2L-1}), & l_i = L/2 \\
            (2^{l_i+L}, 2^{l_i+L-1}), & \text{otherwise},
        \end{cases}
    \end{align}
    for \( i = 0, \cdots, 3 \), and \( x + b_l \) corresponds to section \( i_c = 0 \), while \( x + c_l \) corresponds to section \( i_c = 1 \). 
    The fact that \( b^* = 0 \) can be demonstrated as in Theorem \ref{thr:classicalQC-LDPC-matrix-achieve-girth12-half}.
    Moreover, 
    \begin{align}
        |b^*|
        \le & |2^{L/2} - 2^{L-1} + 2^{L-2} - 2^{L-1} \\
        &+ 2^{3L/2} - 2^{2L-1} + 2^{2L-2} - 2^{2L-1}| \\
        =& |2^{L/2} + 2^{L-2} - 2^L + 2^{3L/2} + 2^{2L-2} - 2^{2L}| \\
        <& 2^L (2^L + 1) \\
        =& P_1.
    \end{align}
\end{proof}